\documentclass[a4paper]{article}
\usepackage[a4paper,top=3cm,bottom=2cm,left=3.5cm,right=3.5cm,marginparwidth=4cm]{geometry}
\usepackage{amsmath,amsfonts}
\usepackage{algorithmic}
\usepackage{algorithm}
\usepackage{array}
\usepackage[caption=false,font=normalsize,labelfont=sf,textfont=sf]{subfig}
\usepackage{textcomp}
\usepackage{stfloats}
\usepackage{url}
\usepackage{verbatim}
\usepackage{graphicx}
\usepackage{cite}
\usepackage{authblk}
\usepackage{mathtools}

\usepackage{mathptmx}
\usepackage{amsthm,amssymb,latexsym}
\usepackage{bbm}
\usepackage{booktabs}
\usepackage{multirow}

\theoremstyle{plain}
\newtheorem{theorem}{Theorem}

\newtheorem{proposition}[theorem]{Proposition}

\theoremstyle{definition}

\newtheorem{definition}[theorem]{Definition}

\newtheorem{problem}{Problem} 

\newtheorem{formulation}{Formulation} 

\theoremstyle{remark}

\newcommand{\indeg}{\mathrm{indeg}}
\newcommand{\outdeg}{\mathrm{outdeg}}
\newcommand{\head}{\mathrm{head}}
\newcommand{\tail}{\mathrm{tail}}

\newif\ifhighlight
\highlighttrue 
\usepackage{xcolor}
\definecolor{mypink}{rgb}{0.9, 0.0, 0.4}
\definecolor{mygreen}{rgb}{0.0,0.5,0.0}
\definecolor{mypurple}{rgb}{0.6,0.0,0.6}
\definecolor{myfuchsia}{rgb}{0.9,0.0,0.9}
\definecolor{mybrown}{rgb}{0.7,0.3,0.3}
\definecolor{mygray}{rgb}{0.6,0.6,0.6}
\definecolor{myltgray}{rgb}{0.85,0.85,0.85}
\definecolor{stnavy}{HTML}{0D1CC9}
\definecolor{teal}{rgb}{0.0, 0.5, 0.5}
\ifhighlight 
\newcommand{\todo}[1]{\textcolor{red}{\small{[ToDo: #1]}}}
\newcommand{\suggested}[1]{\textcolor{myfuchsia}{\small{[Suggested: #1]}}}

\else 
\newcommand{\todo}[1]{} 
\newcommand{\suggested}[1]{} 
\fi

\begin{document}

\title{Finding Tree-Like Substructures in Phylogenetic Networks: ILP Approaches and Their Application}

\author[1]{Takatora Suzuki}
\affil[1]{Department of Applied Mathematics, Faculty of Science and Engineering, Waseda University, Tokyo 169-8555, Japan}

\maketitle

\begin{abstract}
Phylogenetic networks model evolutionary histories that involve reticulate events, but their structural complexity makes them difficult to interpret. 
Extracting their simple substructures both clarifies the evolutionary pathways and quantifies the complexity of the networks themselves. 
For a given rooted almost-binary phylogenetic network, the \textsc{Level Minimization} problem asks for a spanning subgraph that has the same root and leaf-set and whose level is minimum, i.e., which is as close to a tree as possible.
Networks for which the minimum level is zero are known as tree-based networks and can be recognized in linear time. 
However, \textsc{Level Minimization} is NP-hard in general. 
State-of-the-art algorithms rely on exhaustive searches of the solution spaces and hence apply only to networks of limited size. 
In this paper, we propose two methods for \textsc{Level Minimization} using integer linear programming: an exact formulation for finding such a subgraph of level at most one, and a heuristic formulation for the general case. 
Computational experiments confirmed the practicality of both formulations. 
An application to ancestral recombination graphs suggests that the minimum level provides an alternative measure of the topological complexity of an inferred network.
\end{abstract}

\section{Introduction}
Phylogenetic trees represent the branching evolutionary history among biological species, but cannot represent reticulate events such as horizontal gene transfer~\cite{szollHosi2015genome}, hybridization~\cite{goulet2017hybridization} and recombination. 
Phylogenetic networks extend trees so as to accommodate such events and have been used in a wide range of evolutionary analyses~\cite{BAPTESTE2013439,huson2011survey,kong2022classes}. 
This expressive power comes at the cost of structural complexity, which makes networks more difficult to interpret than trees. 
Extracting simple substructures from a network is a natural approach to making it easier to interpret.

An influential formalization of this approach is the notion of a \emph{support tree}, introduced by Francis and Steel~\cite{FrancisSteel-2015-WhichPhylogeneticNetworks}. 
A support tree of a phylogenetic network $N$ is a spanning tree of $N$ with the same root and leaf-set as $N$, and a network containing at least one support tree is called \emph{tree-based}. Tree-based networks have since been studied
extensively~\cite{Zhang-2016-TreeBasedPhylogeneticNetworks,fss2018,Non-binary-tree-based,francis2018tree}. 
A major development in this area was the structure theorem given by Hayamizu~\cite{Hayamizu-2021-StructureTheoremRooted}, which canonically decomposes a rooted almost-binary phylogenetic network~$N$ into certain subgraphs called \emph{maximal zig-zag trails} and characterizes the family of support trees of~$N$. This theorem led to optimal algorithms for a variety of computational problems related to support trees, such as recognition, counting, enumeration and optimization.

However, not every phylogenetic network is tree-based, and hence a support tree does not always exist. 
For such networks, a natural alternative is to extract a \emph{support network}, a spanning subgraph of $N$ with the same root and leaf-set as $N$, and to optimize its closeness to a tree. 
Here, the tree-likeness is often measured by the \emph{level} of a network (see, e.g.,\ \cite{CHOY200593,van2009uniqueness}), an index of how heavily its reticulations are entangled; in particular, level-$0$ networks are exactly the trees.
Suzuki and Hayamizu~\cite{SuzukiEtAl-2025-WhichPhylogeneticNetworks} formulated the \textsc{Level Minimization} problem: given a rooted almost-binary network $N$, find a support network of $N$ with the minimum level among all support networks of $N$.
This minimum value is called the \emph{base level} of $N$.
Building on a characterization of the family of all support networks of $N$ and its subfamilies,
they obtained an exact algorithm and a heuristic algorithm for \textsc{Level Minimization} by exhaustively searching such subfamilies.
However, both require exponential time in the worst case, so that they apply only to networks of limited size. 
The NP-hardness of \textsc{Level Minimization} was conjectured in~\cite{SuzukiEtAl-2025-WhichPhylogeneticNetworks} and then proved by Suzuki~\cite{suzuki2026hardness}.

In this paper, we propose two integer linear programming (ILP) formulations for extracting a support network of small level from a given rooted almost-binary phylogenetic network.
Both encode the structural results of~\cite{SuzukiEtAl-2025-WhichPhylogeneticNetworks} as linear constraints, so that a solver can prune the search space instead of enumerating it exhaustively.
The first is an exact formulation for the problem of finding a support network of level at most one. 
The second is a heuristic formulation for \textsc{Level Minimization}, which minimizes the overlaps among the reticulation cycles of a support network.

We evaluated the practicality of both formulations experimentally using the Gurobi Optimizer~\cite{gurobi}. On randomly generated networks, the first formulation scaled to inputs with up to $800$ reticulations, and the second was faster and far more stable than the exhaustive-search heuristic of~\cite{SuzukiEtAl-2025-WhichPhylogeneticNetworks} while returning support networks of almost the same level. 
We then applied our methods to two ancestral recombination graphs (ARGs)~\cite{wong2024general} inferred from the same \textit{Drosophila melanogaster} sequence data~\cite{kreitman1983nucleotide} by KwARG~\cite{ignatieva2021kwarg} and ARGweaver~\cite{rasmussen2014genome}. 
The two ARGs turned out to have different base levels, which suggests that the base level can be used as an alternative measure of the topological complexity of an inferred network.

The rest of the paper is organized as follows. 
Section~\ref{sec:prelim} gives graph theoretical terminology, defines phylogenetic networks, and recalls the structure theorem from~\cite{Hayamizu-2021-StructureTheoremRooted}. 
Section~\ref{sec:problem} states the problems of interest and briefly reviews relevant results, including the state-of-the-art algorithms in~\cite{SuzukiEtAl-2025-WhichPhylogeneticNetworks}. 
Section~\ref{sec:ILP} 
gives a characterization of level-$1$ networks (Theorem~\ref{prop:level-1 chara}) and  describes the proposed methods (Formulations~\ref{form:level-1 exact} and~\ref{form:level min}).
Section~\ref{sec:exp} presents our computational experiments, in which we evaluate the scalability of Formulation~\ref{form:level-1 exact} and compare Formulation~\ref{form:level min} with Algorithm~2 of~\cite{SuzukiEtAl-2025-WhichPhylogeneticNetworks}. 
Section~\ref{sec:exp3} applies both methods to ARGs inferred from empirical data. 
Section~\ref{sec:conclusion} concludes the paper and discusses directions for future work.

\section{Preliminaries}
\label{sec:prelim}
\subsection{Graph theoretical terminology}\label{subsec:graph}
Throughout this paper, all graphs are finite, simple (i.e.,\ having neither loops nor multiple edges), acyclic directed  graphs, unless otherwise stated. 
For a graph $G$, $V(G)$ and $E(G)$ denote the sets of vertices and edges of $G$, respectively. 
For two graphs $G$ and $H$, $G$ is a \emph{subgraph} of $H$ if both $V(G) \subseteq V(H)$ and $E(G) \subseteq E(H)$ hold, in which case we write $G \subseteq H$ or $H\supseteq G$. 
Two graphs $G$ and $H$ are \emph{isomorphic}, denoted by $G \cong H$, if there exists a bijection $\varphi : V(G) \to V(H)$ such that $(u,v) \in E(G)$ if and only if $(\varphi(u), \varphi(v)) \in E(H)$ for all $u, v \in V(G)$. A subgraph $G$ of $H$ is \emph{proper} if $G\neq H$. A subgraph  $G$ of $H$ is a \emph{spanning} subgraph of $H$ if $V(G)=V(H)$.

Given a graph $G$ and a non-empty subset $S \subseteq E(G)$, the edge-set $S$ is said to \emph{induce the subgraph  $G[S]$ of $G$}, that is, the one whose edge-set is $S$ and whose vertex-set is the set of ends of all edges in $S$. 
For a graph $G$ with $|E(G)| \geq 1$ and a partition  $\{E_1,\dots,E_{d}\}$ of $E(G)$, the collection  $\{G[E_1], \dots, G[E_{d}]\}$ is a  \emph{decomposition} of $G$, and $G$ is said to be \emph{decomposed into} $\{G[E_1], \dots, G[E_{d}]\}$. 

For an edge $e=(u,v)$ of a graph~$G$, $u$ and $v$ are denoted by $\tail(e)$ and $\head(e)$, respectively.
For vertices $u$ and $v$ of $G$, $u$ is a \emph{parent} of $v$, and $v$ is a \emph{child} of $u$, if $(u, v)\in E(G)$. 
For a vertex $v$ of $G$, we define $\delta^-_G(v) = \{e\in E(G)\mid head(e) = v\}$ and $\delta^+_G(v) = \{e\in E(G)\mid tail(e) = v\}$. 
The \emph{in-degree} (resp.\ \emph{out-degree}) of $v$ in $G$, denoted by $\indeg_G(v)$ (resp.\ $\outdeg_G(v)$), is defined by $|\delta^-_G(v)|$ (resp.\ $|\delta^+_G(v)|$). 
Also, the \emph{degree} of $v$ in $G$, denoted by $\deg_G(v)$, is the sum of the in-degree and the out-degree of $v$ in $G$. 
For a graph~$G$, a vertex $v\in V(G)$ is called a \emph{root} of $G$ if $\indeg_G(v) = 0$; a \emph{leaf} of $G$ if $(\indeg_G(v), \outdeg_G(v)) = (1,0)$; 
and a \emph{reticulation} of $G$ if $\indeg_G(v) = 2$. 

Let $G$ be a graph. \emph{Subdividing} an edge $(u,v)$ of $G$ means replacing it with a directed path from $u$ to $v$ of length at least two. \emph{Smoothing} a vertex $v$ with $\indeg_G(v)=\outdeg_G(v)=1$ means suppressing $v$ from $G$, which is the reverse operation of edge subdivision.
A graph $G'$ is called a \emph{subdivision} of $G$ if $G'$ is obtained from $G$ by applying edge subdivisions zero or more times.

For a directed graph $G$, the \emph{underlying undirected graph} of $G$, denoted by $\tilde{G}$, is the undirected graph obtained from $G$ by replacing each edge $(u,v)$ with the undirected edge $\{u,v\}$.
An undirected graph is \emph{connected} if there is a path between every pair of vertices. 
A directed graph is \emph{connected} if its underlying undirected graph is connected.
For a connected simple undirected graph $G$, a \emph{cut vertex} of $G$ is a vertex whose removal disconnects $G$. 
A subgraph $H$ of $G$ is called a \emph{block} of $G$ if $H$ is a maximal connected subgraph of $G$ that contains no cut vertex of $H$. 
In this paper, a \emph{block} of a directed graph~$G$ refers to a block of $\tilde{G}$.
A block is \emph{trivial} if it consists of a single edge or a single vertex; otherwise, it is \emph{non-trivial}.
Let $H$ be a subgraph of a graph $G$ such that $\tilde{H}$ is a cycle. If $H$ has a unique vertex $v$ with $\indeg_H(v) =2$, we call $H$ a \emph{reticulation cycle}, or a \emph{$v$-cycle}, of~$G$.

\subsection{Phylogenetic networks}\label{subsec:phylo net}
Throughout this paper, $X$ represents a non-empty finite set, which can be interpreted as a set of present-day species.
\begin{definition}
  A \emph{rooted almost-binary phylogenetic network (on a leaf-set $X$)} is defined to be a finite simple directed acyclic graph $N = (V,E)$ with the following properties (P1)--(P3):
\begin{description}
    \item[(P1)] $N$ has a unique root $\rho$ and $\rho$ has out-degree 1 or 2;
    \item[(P2)] the set of leaves of $N$ is identical to $X$; and
    \item[(P3)] every vertex $v\in V\setminus (X\cup \{\rho\})$ satisfies $\indeg_N(v)\in \{1,2\}$ and $\outdeg_N(v)\in \{1,2\}$.
\end{description}
A rooted almost-binary phylogenetic network satisfying (P3') below is called a rooted \emph{binary} phylogenetic network.
\begin{description}
  \item[(P3')] every vertex $v\in V\setminus (X\cup \{\rho\})$ satisfies $(\indeg_N(v),\outdeg_N(v)) \in \{(1,2),(2,1)\}$.
\end{description}
\end{definition}

Let $N$ be a rooted almost-binary phylogenetic network on~$X$ and let $G$ be a spanning subgraph of $N$.
If $G$ is a rooted almost-binary phylogenetic network on $X$, we call $G$ a \emph{support network} of~$N$. 
A support network $G$ of $N$ is called \emph{minimal} if no support network of $N$ is its proper subgraph, and is called \emph{minimum} if it has the minimum number of edges among all support networks of $N$.
A support network $G$ of $N$ is called a \emph{support tree} if $G$ contains no reticulation.
If $N$ has a support tree, $N$ is called \emph{tree-based}.

Since every support network of $N = (V, E)$ has the same vertex set $V$, we shall identify a support network $G = (V, E')$ of $N$ with its edge-set $E' \subseteq E$ throughout this paper. 
We denote the families of all, minimal, and minimum support networks of $N$ by $\mathcal{A}_N, \mathcal{B}_N$ and $\mathcal{C}_N$, respectively.
We note that a support network of $N$ always exists since $N$ itself is a support network of $N$.

For a directed graph $G$, 
$\mathrm{level}(G)$ denotes the \emph{level} of $G$, i.e., the maximum number of reticulations contained in a block of $G$. 
The \emph{base level} of a rooted almost-binary phylogenetic network $N$, denoted by $\mathrm{level}^\ast(N)$, is the minimum value of $\mathrm{level}(G)$ over all support networks $G$ of $N$. 
If $k\geq 0$ is the base level of $N$, then $N$ is called \emph{level-$k$-based}. 
We note that $N$ is tree-based if and only if $N$ is level-$0$-based, since level-$0$ phylogenetic networks are exactly the phylogenetic trees.

\subsection{Structure theorem for rooted almost-binary phylogenetic networks}\label{subsec:structure theorem}
Here we recall relevant materials from~\cite{Hayamizu-2021-StructureTheoremRooted}.
Let $N$ be a rooted almost-binary phylogenetic network on $X$. 
A connected subgraph $Z$ of $N$ with $m\geq 1$ edges is called a \emph{zig-zag trail} if there exists a permutation $(e_1, \dots, e_m)$ of $E(Z)$ such that either $\head(e_i) = \head(e_{i+1})$ or $\tail(e_i) = \tail(e_{i+1})$ holds for all $i\in[1, m-1]$. 
In this paper, we often identify $Z$ with the permutation $(e_1, \dots, e_m)$.
A zig-zag trail $Z$ is \emph{maximal} if no zig-zag trail contains $Z$ as a proper subgraph.
We write each edge $(v_i, v_j)$ as $v_i > v_j$ or $v_j < v_i$ to concisely express $Z$ as $v_0 > v_1 < v_2 \dots > v_{m-1} < v_m$ or its reverse. 

Every maximal zig-zag trail is classified into one of the four types. 
A maximal zig-zag trail $Z$ is called a \emph{crown} if $Z$ has an even number $m$ of edges and can be written in the cyclic form $v_0 > v_1 < \dots > v_{m-1} < v_m = v_0$.
A \emph{fence} is a maximal zig-zag trail that is not a crown. 
An \emph{M-fence} is a fence with an even number $m$ of edges and of the form $v_0 < v_1 > \dots < v_{m-1} > v_m \neq v_0$. 
A \emph{W-fence} is a fence with an even number $m$ of edges and of the form $v_0 > v_1 < \dots > v_{m-1} < v_m \neq v_0$. 
An \emph{N-fence} is a fence with an odd number $m$ of edges and of the form $v_0 > v_1 < \dots < v_{m-1} > v_m$. 

As in the approach taken in~\cite{Hayamizu-2021-StructureTheoremRooted}, we often consider a maximal zig-zag trail $Z$ as a sequence $(e_1,\dots,e_{|E(Z)|})$ of edges, ordered according to their appearance in the trail.
For any maximal zig-zag trail $Z = (e_1, \dots,  e_{|E(Z)|})$ in $N$, any subset $S$ of $E(Z)$  is specified by a $0$-$1$ sequence $\langle b_1\; b_2\; \dots \; b_{|E(Z)|}\rangle$, where $b_i = 1$ if $e_i \in S$  and $b_i = 0$ otherwise. 

Using the terminology introduced above, we now state the part of the structure theorem of Hayamizu~\cite{Hayamizu-2021-StructureTheoremRooted} that we use in this paper.
\begin{theorem}%
  [\cite{Hayamizu-2021-StructureTheoremRooted}, see Figure~\ref{fig:zig-zag} for example]
  \label{thm:zig-zag trail decomposition}
  Let $N$ be a rooted almost-binary phylogenetic network on $X$. Then, there exists a unique decomposition $\mathcal{Z} = \{Z_1, \dots, Z_d\}$ of $N$, where each $Z_i$ is a maximal zig-zag trail of $N$. Moreover, this decomposition can be computed in $\Theta(|E(N)|)$ time. 
\end{theorem}
By the uniqueness in Theorem~\ref{thm:zig-zag trail decomposition}, we may call the decomposition $\mathcal{Z}$ therein the \emph{maximal zig-zag trail decomposition} of $N$. 
We denote by $\mathcal{Z}_c\subset \mathcal{Z}$ and $\mathcal{Z}_f\subseteq\mathcal{Z}$ the sets of all crowns and of all fences of $N$, respectively.

\begin{figure}[t]
  \centering
  \includegraphics[width=0.45\textwidth]{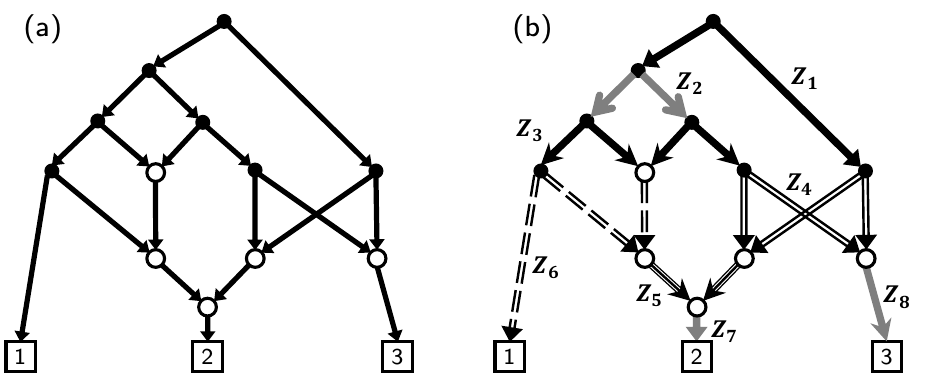}
  \caption{(a) A rooted binary phylogenetic network~$N$ on $X = \{1, 2, 3\}$. (b) The maximal zig-zag trail decomposition~$\mathcal{Z} = \{Z_1, \dots, Z_8\}$ of $N$, where $Z_1, Z_2$ and $Z_3$ are M-fences, $Z_4$ is a crown, $Z_5$ is a W-fence, and $Z_6, Z_7$ and $Z_8$ are N-fences. Each maximal zig-zag trail is depicted in different arrow styles.}
  \label{fig:zig-zag}
\end{figure}

\section{Problems and known results}\label{sec:problem}
In this section, we formally describe the problems that we focus on and summarize what is known about them.
\begin{problem}[\textsc{Level Minimization},~\cite{SuzukiEtAl-2025-WhichPhylogeneticNetworks}]
  \label{prob:level min}
  Given a rooted almost-binary phylogenetic network~$N$ on $X$, output a value $k = \mathrm{level}^\ast(N)$ and a support network of $N$ of level $k$.
\end{problem}
The following Problem~\ref{prob:find level-1} may be viewed as a partial version of Problem~\ref{prob:level min}.
\begin{problem}
  \label{prob:find level-1}
  Given a rooted almost-binary phylogenetic network~$N$ on $X$, output a support network $G\in \mathcal{A}_N$ with $\mathrm{level}(G)\leq 1$ if it exists, and `NO' otherwise.
\end{problem}
Suzuki~\cite{suzuki2026hardness} has shown that both problems are NP-hard; this intractability motivates the present work.
However, recall that $N$ is tree-based if and only if $\mathrm{level}^\ast(N)=0$. 
The following proposition, together with Theorem~\ref{thm:zig-zag trail decomposition}, yields an $O(|E(N)|)$-time algorithm for deciding whether $\mathrm{level}^\ast(N) = 0$ holds. 
Equivalent results appear in~\cite{Zhang-2016-TreeBasedPhylogeneticNetworks,Non-binary-tree-based}.

\begin{proposition}[\cite{Zhang-2016-TreeBasedPhylogeneticNetworks,Non-binary-tree-based,Hayamizu-2021-StructureTheoremRooted}]\label{prop:tree-based}
  Let~$N$ be a rooted almost-binary phylogenetic network on $X$ and let $\mathcal{Z}$ be the maximal zig-zag trail decomposition of $N$. 
  Then, $N$ is tree-based if and only if no element of $\mathcal{Z}$ is a W-fence. 
\end{proposition}

For the intractable cases, the state-of-the-art algorithms for Problem~\ref{prob:level min} rest on the structural results of Suzuki and Hayamizu~\cite{SuzukiEtAl-2025-WhichPhylogeneticNetworks}, who characterized each of $\mathcal{A}_N$, $\mathcal{B}_N$ and $\mathcal{C}_N$ as a direct product over the maximal zig-zag trails of $N$. We recall only the characterization of $\mathcal{B}_N$, which we use in Section~\ref{sec:ILP}, and refer the reader to~\cite{SuzukiEtAl-2025-WhichPhylogeneticNetworks} for the other two.
\begin{proposition}%
  [\cite{SuzukiEtAl-2025-WhichPhylogeneticNetworks}]
  \label{thm:bijection.ABC}
  Let $N$ be a rooted almost-binary phylogenetic network on~$X$ and let $\mathcal{Z} = \{Z_1, \dots, Z_d\}$ be the maximal zig-zag trail decomposition of $N$.
  Then the set $\mathcal{B}_N$ is characterized by $\mathcal{B}_N = \prod_{i=1}^d \mathcal{S}_{\mathcal{B}} (Z_i)$, where $(Z_1, \dots, Z_d)$ is an arbitrary ordering of~$\mathcal{Z}$ and each $\mathcal{S}_{\mathcal{B}} (Z_i)$ is expressed by~\eqref{eq:sequence for minimal support network}. Here, we let $m_i = |E(Z_i)|$.
  \begin{align}
  \mathcal{S}_{\mathcal{B}}(Z_i) = 
  \begin{cases}
    \left\{ \langle b_1\cdots b_{m_i}\rangle \;\middle|\; 
      \begin{array}{@{}l@{}}
        b_1 = b_{m_i} = 1, \text{ and} \\
        \text{no $\langle 00\rangle$ or $\langle 111\rangle$ occurs}
      \end{array}
    \right\}    \hfill \forall\  Z_i\in \mathcal{Z}_f \\[3ex]
    \left\{ \langle b_1 \cdots b_{m_i}\rangle \;\middle|\; 
      \begin{array}{@{}l@{}}
        \text{no $\langle 00\rangle$ or $\langle 111 \rangle$ occurs} \\
        \text{in any circular ordering}
      \end{array}
    \right\} \hfill \forall\  Z_i\in \mathcal{Z}_c
  \end{cases} \label{eq:sequence for minimal support network}
\end{align}
\end{proposition}

They also proved that searching $\mathcal{B}_N$ suffices for Problem~\ref{prob:level min}.

\begin{proposition}[\cite{SuzukiEtAl-2025-WhichPhylogeneticNetworks}]
  \label{lem:BN search correctness}
  For any rooted almost-binary phylogenetic network $N$ on $X$, there exists at least one element $G\in \mathcal{B}_N$ that satisfies $\mathrm{level}(G) = \mathrm{level}^\ast(N)$.
\end{proposition}

By these results, Suzuki and Hayamizu~\cite{SuzukiEtAl-2025-WhichPhylogeneticNetworks} obtained two algorithms for Problem~\ref{prob:level min}: an exact one (Algorithm~1) that exhaustively searches $\mathcal{B}_N$ (its correctness rests on Proposition~\ref{lem:BN search correctness}), and a heuristic one (Algorithm~2) that exhaustively searches $\mathcal{C}_N$. 
The algorithms run in $O(|E(N)| \cdot |\mathcal{B}_N|)$ time and $O(|E(N)| \cdot |\mathcal{C}_N|)$ time, respectively. Both runtimes are exponential in $|E(N)|$. 
Problem~\ref{prob:find level-1} can also be solved by exhaustively searching $\mathrm{B}_N$ until a support network of level at most one is found, but no algorithm specialized for this problem has been developed so far.

\section{ILP formulations for finding support network of small level}\label{sec:ILP}
In this section, we propose two integer linear programming (ILP) formulations. 
The former exactly solves Problem~\ref{prob:find level-1}, and the latter provides a heuristic solution for Problem~\ref{prob:level min}.
In what follows, let $N = (V,E)$ be a rooted almost-binary phylogenetic network on $X$, let $R = \{r_1, \dots, r_M\}$ be the set of all reticulations of $N$, where $M = |R|$. 
Let $\mathcal{Z} = \{Z_1, \dots, Z_d\}$ be the maximal zig-zag trail decomposition of $N$, where each $Z_i$ is denoted by $(e_{i,1}, \dots, e_{i,m_i})$ using $m_i = |E(Z_i)|$. 

Moreover, we assume without loss of generality that $N$ contains no vertex $v$ with $\indeg_N(v) = \outdeg_N(v) = 2$.
Indeed, given such a vertex $v$, one can replace $v$ with two new vertices $v_{\mathrm{in}}$ and $v_{\mathrm{out}}$ joined by a new edge $(v_{\mathrm{in}}, v_{\mathrm{out}})$, reattaching the two incoming edges of $v$ to $v_{\mathrm{in}}$ and the two outgoing edges of $v$ to $v_{\mathrm{out}}$.
This operation preserves the base level.

\subsection{An exact ILP formulation for Problem~\ref{prob:find level-1}}
We first give a characterization of level-$1$ networks that is better suited for description with inequalities than the definition via the notion of a block. 
\begin{theorem}
  \label{prop:level-1 chara}
  Let $G$ be a rooted almost-binary phylogenetic network on~$X$ without degree-four vertices and let $R$ be the set of all reticulations of $G$.
  Then $\mathrm{level}(G)\leq 1$ holds if and only if there exists a family $\{C_r\}_{r\in R}$ of pairwise vertex-disjoint subgraphs of $G$ such that each $C_r$ is an $r$-cycle of $G$.
\end{theorem}

\begin{proof}
  Suppose first that $\mathrm{level}(G)\leq 1$. 
Every vertex of a non-trivial block $B$ has degree at least two in $B$, while every vertex of $G$ has degree at most three; hence no vertex of $G$ lies in two non-trivial blocks, and the non-trivial blocks of $G$ are pairwise vertex-disjoint. 
By the assumption on the level, each of them contains exactly one reticulation $r$ of $G$ and is therefore an $r$-cycle. 
Since the two incoming edges of each reticulation lie in a common non-trivial block, the non-trivial blocks of $G$ give the required family.

Conversely, suppose that such a family $\{C_r\}_{r\in R}$ exists and that $\mathrm{level}(G) = k \geq 2$. 
Let $B$ be a block of $G$ containing $k$ reticulations $r_1,\dots, r_k$ of $G$. We note that $|E(B)| = |V(B)| + k - 1$ holds by counting the edges of $B$ by their heads. 
Fix $i\in [1,k]$. 
Since $C_{r_i}$ is biconnected and contains $r_i$, it is a subgraph of $B$. 
Let $u\in V(C_{r_i})$ and let $j\in [1,k]\setminus\{i\}$. 
As $B$ is biconnected, Menger's theorem (Theorem~4.19 in~\cite{chartrand2024graphs})  yields two internally disjoint undirected paths from $u$ to $r_j$ in the underlying undirected graph of $B$. 
Since $r_j\notin V(C_{r_i})$, each of them leaves $C_{r_i}$ at a vertex incident with an edge of $B$ not in $C_{r_i}$, and these two vertices are distinct because every vertex of $C_{r_i}$ is incident with exactly two edges of $C_{r_i}$ and has degree at most three in $G$. 
Hence at least two vertices of $C_{r_i}$ have degree three in $B$. 
As $C_{r_1},\dots, C_{r_k}$ are pairwise vertex-disjoint, $B$ contains at least $2k$ vertices of degree three, and the handshaking lemma gives $|E(B)| \geq \bigl(3\cdot 2k + 2(|V(B)| - 2k)\bigr)/2 = |V(B)| + k$. 
This contradicts the equality obtained above. 
\end{proof}

Here we describe an ILP formulation for finding a support network of $N = (V, E)$ of level at most one:

\begin{formulation}\label{form:level-1 exact}
  \setlength{\abovedisplayskip}{2pt}\setlength{\abovedisplayshortskip}{0pt}%
\setlength{\belowdisplayskip}{2pt}\setlength{\belowdisplayshortskip}{0pt}%
\begin{align}
   \min\ & 0 \notag\\
  \text{s.t.}\ 
  & s_{e_{i,1}} = s_{e_{i,m_i}} = 1 & \forall\  Z_i \in \mathcal{Z}_f \label{eq:fence terminal}\\
  & s_{e_{i,j}} + s_{e_{i,j+1}} \geq 1 & \forall\ Z_i \in \mathcal{Z}_f,\ j\in [1, m_i-1] \label{eq:fence two consec}\\
  & s_{e_{i,j}} + s_{e_{i,j+1}} + s_{e_{i,j+2}} \leq 2 & \forall\ Z_i \in \mathcal{Z}_f,\ j\in [1, m_i-2]\label{eq:fence three consec}\\
  & s_{e_{i,j}} + s_{e_{i,j+1}} \geq 1 & \forall\ Z_i \in \mathcal{Z}_c,\ j\in [1, m_i]\label{eq:crown two consec}\\
  & s_{e_{i,j}} + s_{e_{i,j+1}} + s_{e_{i,j+2}} \leq 2 & \forall\ Z_i \in \mathcal{Z}_c,\ j\in [1, m_i]\label{eq:crown three consec}\\
  & z_e\leq s_e & \forall\  e\in E\label{eq:ze leq se}\\
  & c_v + t_v + p_v \leq 1 & \forall\  v\in V \label{eq:three in one}\\
  & \sum_{e\in \delta^-_N(v)} z_e = c_v + 2t_v & \forall\ v\in V \label{eq:indeg}\\
  & \sum_{e\in \delta^+_N(v)} z_e = c_v + 2p_v & \forall\ v\in V \label{eq:outdeg}\\
  & t_v=\sum_{e\in\delta^-_N(v)}s_e-1 & \forall\ v\in R \label{eq:ret count}\\
  & t_v=0 & \forall\  v\in V\setminus R \label{eq:non-ret}\\
  & |\lambda_u - \lambda_v| \leq M(1-z_{(u,v)}) & \forall\ (u, v)\in E \label  {eq:label carryout}\\
  & |\lambda_{r_i} - i| \leq M(1-t_{r_i}) & \forall\ i \in [1,M]\label{eq:ret label force}\\
  & s_e, t_v, z_e, c_v, p_v\in \{0,1\} & \forall\ e\in E,\ v\in V \notag\\
  & \lambda_{v}\in \mathbb{N}\cap [0,M] & \forall\ v\in V\notag
\end{align}
Note that in constraints~\eqref{eq:crown two consec} and~\eqref{eq:crown three consec}, the second subscripts are taken modulo $m_i$ for each $Z_i$.
\end{formulation}

Here, $s_e$ indicates whether an edge $e$ belongs to a support network $G$ of $N$, and $z_e$ whether $e$ lies in one of the chosen reticulation cycles of $G$. 
Let $E_z = \{e\in E\mid z_e = 1\}$ and $G_z = N[E_z]$. 
For each $v\in V$, the variables $c_v$, $t_v$ and $p_v$ indicate whether $(\indeg_{G_z}(v), \outdeg_{G_z}(v))$ equals $(1,1)$, $(2,0)$ and $(0,2)$, respectively.
Finally, $\lambda_v$ is set to $i$ if $v$ lies on the $r_i$-cycle of $G_z$.

We now explain why Formulation~\ref{form:level-1 exact} correctly finds a level-1 support network of $N$.
Let $(s, t, z, c, p, \lambda)$ be a feasible solution, let
$E' = \{e\in E\mid s_e = 1\}$, and let $G = N[E']$.
Constraints~\eqref{eq:fence terminal}--\eqref{eq:crown three consec} express
\eqref{eq:sequence for minimal support network} exactly, and hence are
equivalent to $G\in\mathcal{B}_N$ by Proposition~\ref{thm:bijection.ABC}.

The remaining constraints make $G_z$ the union of a family $\{C_r\}_{r\in R'}$ as in Theorem~\ref{prop:level-1 chara}, where $R'$ is the set of all reticulations of $G$. 
Constraint~\eqref{eq:ze leq se} ensures $G_z \subseteq G$. 
For $v\in V$, constraint~\eqref{eq:three in one} leaves exactly four possibilities for $(c_v, t_v, p_v)$, namely $(0,0,0)$, $(1,0,0)$, $(0,1,0)$ or $(0,0,1)$, for which constraints~\eqref{eq:indeg} and~\eqref{eq:outdeg} give $(\indeg_{G_z}(v), \outdeg_{G_z}(v)) = (0,0)$, $(1,1)$, $(2,0)$ or $(0,2)$, respectively. 
Hence every vertex of $G_z$ has degree two, so that the underlying undirected graph of each connected component of $G_z$ is a cycle, and $t_v = 1$ holds exactly at the vertices of in-degree two in $G_z$. 
Constraints~\eqref{eq:ret count} and~\eqref{eq:non-ret} force $t_v = 1$ precisely when $v$ is a reticulation of $G$.
Constraints~\eqref{eq:label carryout} and~\eqref{eq:ret label force} ensure that distinct reticulations of $G$ lie on distinct components: the former gives $\lambda_u = \lambda_v$ whenever $(u,v)\in E_z$, so that $\lambda$ is constant on each component of $G_z$, while the latter pins $\lambda_{r_i} = i$ at every reticulation $r_i$ of~$G$. 

Therefore, each component of $G_z$ is an $r$-cycle of $G$ for some reticulation $r$, and these cycles are pairwise vertex-disjoint. By Theorem~\ref{prop:level-1 chara}, $G$ is of level at most one.
Conversely, if $G \in \mathcal{B}_N$ is of level at most one, then Theorem~\ref{prop:level-1 chara} guarantees that such $r$-cycles exist and their union, which is empty when $G$ is a tree, yields a feasible solution. 
Hence $\mathrm{level}^\ast(N)\leq 1$ holds if and only if Formulation~\ref{form:level-1 exact} is feasible, in which case $G$ is a desired support network of $N$.

\subsection{A heuristic ILP formulation for Problem~\ref{prob:level min}}
Let $N$ be a non-tree-based network and let $G\in \mathcal{B}_N$.
By Theorem~\ref{prop:level-1 chara}, an element of $\mathcal{B}_N$ is of level at most one exactly when it has pairwise vertex-disjoint $r$-cycles, one for each of its reticulations $r$. 
Even when no such element of $\mathcal{B}_N$ exists, if we choose an element $G\in\mathcal{B}_N$ and $r$-cycles of $G$ so that they share as few edges as possible, then distinct reticulations tend to lie in distinct blocks of $G$, and hence the level is expected to be small.
For instance, in Figure~\ref{fig:overlap}, the cycles depicted in the level-$3$ support network $G_1$ of $N$ share five edges, whereas those in the level-$2$ support network $G_2$ of $N$ share only two, the minimum over all support networks of $N$ and all choices of their cycles. 
Formulation~\ref{form:level min} below implements this idea and works as a heuristic for Problem~\ref{prob:level min}.

\begin{formulation}\label{form:level min}
    \setlength{\abovedisplayskip}{2pt}\setlength{\abovedisplayshortskip}{0pt}%
\setlength{\belowdisplayskip}{2pt}\setlength{\belowdisplayshortskip}{0pt}%
\begin{align}
   \min\ & \sum_{e\in E} q_e &\label{eq:obj,heur}\\
  \text{s.t.}\ 
  & \eqref{eq:fence terminal}\text{--}\eqref{eq:crown three consec}\notag\\
  & z_{e,r}\leq s_e & \forall\ e\in E,\ r\in R\label{eq:ze leq se,heur}\\
  & c_{v,r} + p_{v,r} + \mathbf{1}_{v=r}t_r \leq 1 & \forall\ v\in V,\ r\in R\label{eq:three in one,heur}\\
  & \sum_{e\in \delta^-_N(v)} z_{e,r} = c_{v,r} + 2\cdot \mathbf{1}_{v=r}t_r & \forall\ v\in V,\ r\in R \label{eq:indeg,heur}\\
  & \sum_{e\in \delta^+_N(v)} z_{e,r} = c_{v,r} + 2p_{v,r} & \forall\ v\in V,\ r\in R\label{eq:outdeg,heur}\\
  & t_r=\sum_{e\in\delta^-_N(r)}s_e-1 & \forall\ r\in R \label{eq:ret count,heur}\\
  & q_e \geq \sum_{r\in R} z_{e, r} -1& \forall\ e\in E\label{eq:count overlap,heur}\\
  & s_e, z_{e,r}\in \{0,1\} & \forall\ e\in E,\ r\in R \notag\\
  & t_r,c_{v,r}, p_{v,r} \in \{0,1\} & \forall\ v\in V,\ r\in R\notag\\
  & q_e\in \mathbb{N}\cap [0,M] & \forall\ e\in E\notag
\end{align}
\end{formulation}
Here, $\mathbf{1}_{v=r}$ represents the indicator function: it returns $1$ if $v=r$ holds, and returns $0$ otherwise. 

In Formulation~\ref{form:level-1 exact}, the variables $z_e$, $c_v$ and $p_v$ describe the single subgraph $G_z$, whose connected components are the $r$-cycles. 
Since we no longer require the disjointness, Formulation~\ref{form:level min} instead keeps track of each $r$-cycle separately by replicating these variables and the constraints on them for every reticulation $r\in R$. 
Indeed, constraints~\eqref{eq:ze leq se,heur}--\eqref{eq:ret count,heur} are copies of~\eqref{eq:ze leq se}--\eqref{eq:ret count}, and hence make the subgraph described by $z_{\cdot,r}$, $c_{\cdot,r}$ and $p_{\cdot,r}$ an $r$-cycle of $G$ for each $r\in R$. 
Constraint~\eqref{eq:count overlap,heur} lets $q_e$ give the lower bound of the number of cycles containing $e$ minus one, and the objective function~\eqref{eq:obj,heur} minimizes the total overlap.

\begin{figure*}[t]
  \centering
  \includegraphics[width=0.8\textwidth]{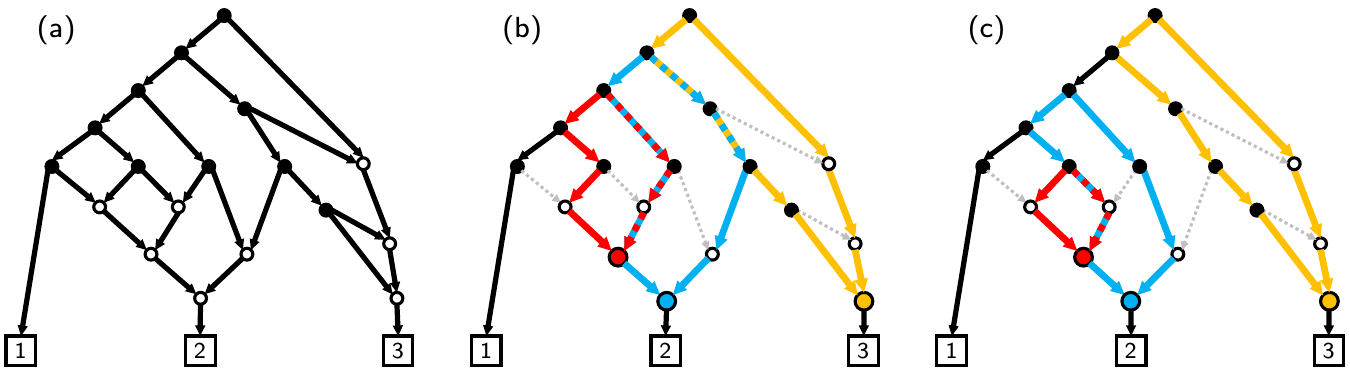}
\caption{(a) A rooted binary phylogenetic network $N$ on $X = \{1,2,3\}$. (b) A level-$3$ support network $G_1$ of $N$. (c) A level-$2$ support network $G_2$ of $N$. 
Support networks are drawn in solid lines, and each chosen $r$-cycle is drawn in the same color as its reticulation $r$. 
}\label{fig:overlap}
\end{figure*}

\section{Experimental Results}\label{sec:exp}
In this section, we apply each of the formulations described in the previous section to randomly generated networks in order to verify their performance. 
All the implementations were coded in Python (version 3.12.2) by relying on Gurobi Optimizer~\cite{gurobi} (version 13.0.2) for ILP.
We refer the runtime of a formulation to the total time spent on preprocessing (i.e.\ computing the maximal zig-zag trail decomposition), building the model and solving it using Gurobi.
All experiments were conducted on a MacBook Pro (2023) with an Apple M3 Max chip, 36\ GB of memory, and macOS Tahoe (version 26.5.2).
The source code, the test data, the program used to generate the data, and the detailed results are available at \url{https://github.com/takatora-suzuki/LevelMinimizationILP}.

\subsection{Experiment 1: scalability of Formulation~\ref{form:level-1 exact}}\label{sec:exp1}
First, we evaluated the scalability of Formulation~\ref{form:level-1 exact} using randomly generated
networks with $8$ leaves.
We generated $100$ rooted binary phylogenetic networks with $n=8$ leaves and $r$ reticulations, for each $r\in\{100,200,\dots,800\}$. Each of them therefore has $3r+14$ edges.
Each instance was generated by first sampling a random level-$2$ network $G$ on $8$ leaves containing $r/5$ level-$1$ blocks and two level-$2$ blocks, and then repeatedly subdividing two randomly chosen edges of a subdivision of $G$, chosen so as not to create a directed cycle, and adding an edge between the two new vertices until the network had $r$ reticulations.

\begin{table*}[t]
\caption{Runtimes for solving Formulation~\ref{form:level-1 exact} on randomly generated rooted binary phylogenetic networks with $8$ leaves and $r \in \{100, 200, \ldots, 800\}$ reticulations. For each $r$, the number of instances and the average, minimum and maximum runtime in seconds are shown separately for YES and NO instances.}
\label{tab:exp1}
\centering
\small
\vspace{1em}
\begin{tabular}{r rrrr rrrr}\hline\hline
 & \multicolumn{4}{c}{YES instances} & \multicolumn{4}{c}{NO instances} \\
\cline{2-5}\cline{6-9}
$r$ & \#\, & avg. & min. & max. & \#\, & avg. & min. & max. \\
\hline
$100$ & 72 & 0.0507 &	0.0314 & 0.1979	& 28 & 0.0415	& 0.0265 & 0.2337\\
$200$ & 56 & 0.2797	& 0.0613 & 1.8659	& 44 & 0.5139	& 0.0495 & 4.8171\\
$300$ & 49 & 1.2586	& 0.1113 & 10.4336 & 51	& 1.4255 & 0.0708	& 8.1404\\
$400$ & 55 & 3.7087	& 0.1678 & 27.7075 & 45	& 3.6673 & 0.1011	& 15.3953 \\
$500$ & 52 & 5.8933	& 0.2074 & 27.2103 & 48 & 6.0797& 0.1229 & 30.0221\\
$600$ & 59 & 11.3215 & 0.5848	& 80.0012 & 41 & 6.4172	& 0.1521 & 43.9449\\
$700$ & 52 & 19.6209 & 0.3261	& 241.2621 & 48 & 21.4806	& 0.1763 & 403.8290\\
$800$ & 65 & 25.6876 & 0.3742	& 417.2249 & 35	& 18.7970	& 0.1946 & 187.7414\\
\hline
\end{tabular}
\end{table*}

Table~\ref{tab:exp1} shows the average, minimum and maximum runtime of Formulation~\ref{form:level-1 exact}, separately for the YES and the NO instances. 
Both classes occurred in every setting, with $49$--$72$ YES instances, and even networks with $r=800$ reticulations were solved in less than 26 seconds on average. 
A log-log regression shows that the mean runtime grows in proportion to $|E(N)|^{\alpha}$ with $\alpha\approx 3.15$ for YES instances and $\alpha\approx 2.98$ for NO instances ($R^2 > 0.98$ in both cases). 
Thus, although Problem~\ref{prob:find level-1} is NP-complete, Formulation~\ref{form:level-1 exact} exhibited polynomial growth on this particular family of instances.
The minimum runtime over the NO instances grew roughly linearly with $r$, indicating that a fair number of NO instances were resolved almost as soon as the model was built.

\subsection{Experiment 2: performance of Formulation~\ref{form:level min}}\label{sec:exp2}
We evaluated Formulation~\ref{form:level min} by comparing it with Algorithm~2 of~\cite{SuzukiEtAl-2025-WhichPhylogeneticNetworks}. We reused the dataset on which the algorithm was evaluated in~\cite{SuzukiEtAl-2025-WhichPhylogeneticNetworks}, which consists of randomly generated rooted binary phylogenetic networks with $8$ leaves and $r\in\{1,\dots,36\}$ reticulations, with $25$ networks per setting.
Their true base levels, computed by the exact algorithm in~\cite{SuzukiEtAl-2025-WhichPhylogeneticNetworks}, are also available.
Among them we used only the non-tree-based networks, since tree-based networks can be recognized in linear time by Theorem~\ref{thm:zig-zag trail decomposition} and Proposition~\ref{prop:tree-based}. 
All $25$ networks were tree-based for $r\in\{1,2,3\}$, and these settings were therefore excluded.
On these inputs, we compared the runtimes of the two methods and the accuracy of the base level they estimated.

\begin{figure}[h]
  \centering
  \includegraphics[width=0.8\textwidth]{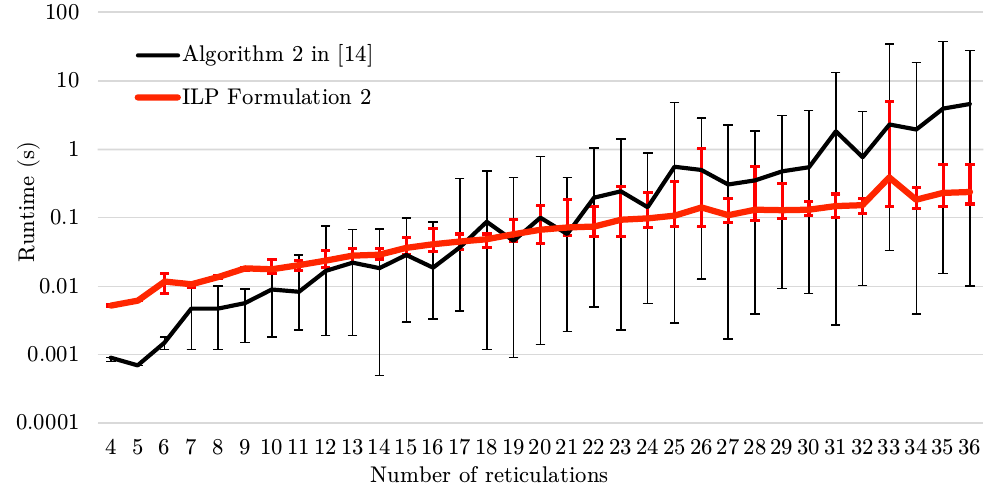}
  \caption{Runtimes of Formulation~\ref{form:level min} (red line) and Algorithm 2 in~\cite{SuzukiEtAl-2025-WhichPhylogeneticNetworks} (black line) on the non-tree-based networks of the dataset of~\cite{SuzukiEtAl-2025-WhichPhylogeneticNetworks}. Runtime is averaged over all samples in each setting, and error bar indicates the minimum and the maximum.}
  \label{fig:exp2 runtime}
\end{figure}

Figure~\ref{fig:exp2 runtime} compares the average, minimum, and maximum runtime of Formulation~\ref{form:level min} and Algorithm~2 of~\cite{SuzukiEtAl-2025-WhichPhylogeneticNetworks}. 
For $r$ up to around $20$, Formulation~\ref{form:level min} was slower than Algorithm~2, but this relation was reversed for larger $r$: around $r=35$, Formulation~\ref{form:level min} was more than ten times faster.
The two methods also differ markedly in the stability of their runtimes. 
The size of $\mathcal{C}_N$ can differ by orders of magnitude between networks with the same number of reticulations, so the cost of the exhaustive search of Algorithm~2 varies accordingly. 
In Formulation~\ref{form:level min}, by contrast, the number of constraints is essentially independent of the decomposition, as it is determined by $|V(N)|$, $|E(N)|$ and $|R|$ except for the number of constraints~\eqref{eq:fence terminal}--\eqref{eq:crown three consec}, which equals $2|E(N)| - |\mathcal{Z}_f|$. 
While the runtime of an integer program is not determined by its size alone, the runtimes we observed varied little among networks with the same number of reticulations.

\begin{figure}[h]
  \centering
  \includegraphics[width=0.7\textwidth]{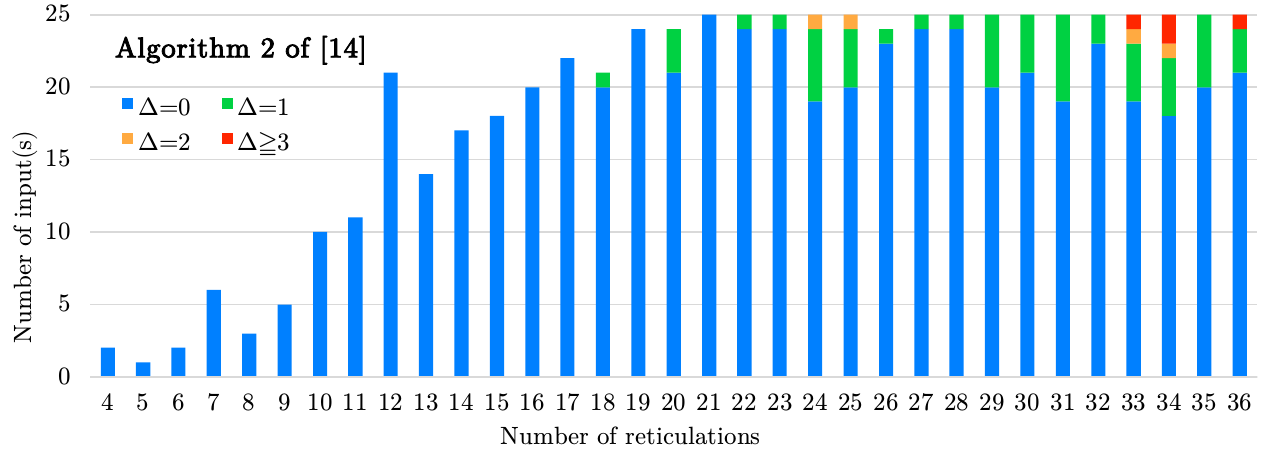}\\
  \includegraphics[width=0.7\textwidth]{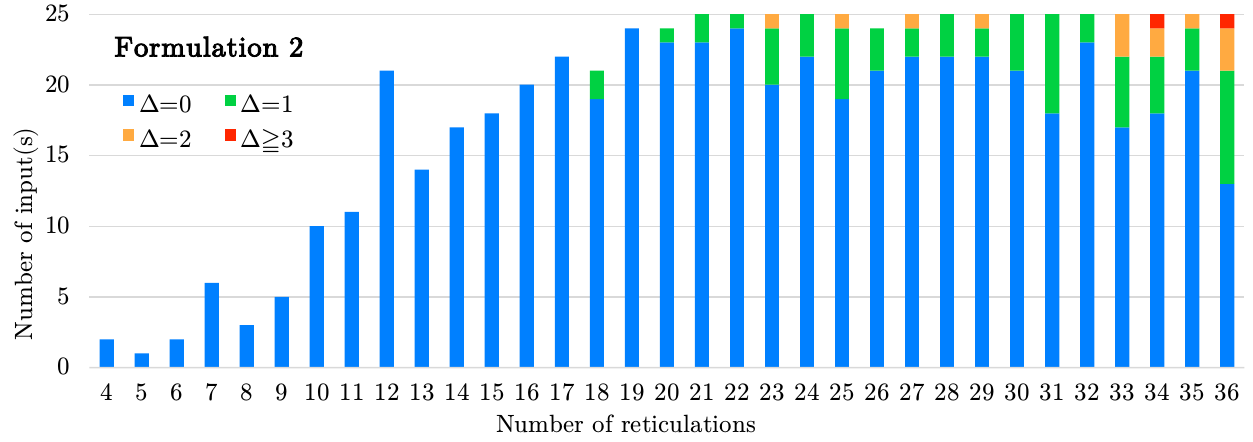}
  \caption{Accuracy of Algorithm~2 of~\cite{SuzukiEtAl-2025-WhichPhylogeneticNetworks} (top) and accuracy of Formulation~\ref{form:level min} (bottom) on the non-tree-based networks of the dataset of~\cite{SuzukiEtAl-2025-WhichPhylogeneticNetworks}. 
  In each figure, each bar shows the number of inputs with difference $\Delta$ between the output and the actual base level: $\Delta = 0$ (blue), $\Delta = 1$ (green), $\Delta = 2$ (orange), $\Delta \geq 3$ (red).}
  \label{fig:exp2 acc}
\end{figure}

Figure~\ref{fig:exp2 acc} shows the deviation $\Delta$, the difference between the level of the support network returned by each method and the base level of the input.
Formulation~\ref{form:level min} attained the base level slightly less often than Algorithm~2 of~\cite{SuzukiEtAl-2025-WhichPhylogeneticNetworks}, and the difference was largest at $r=36$, where the two methods were exact for $13$ inputs and $21$ inputs, respectively. Even there, Formulation~\ref{form:level min}
still returned a support network of minimum level for more than half of the inputs. 
Moreover, the number of inputs with $\Delta\geq 3$ was reduced from four to two compared with Algorithm~2 of~\cite{SuzukiEtAl-2025-WhichPhylogeneticNetworks}, so that the loss in accuracy was confined to small deviations.

\section{Case study: application to phylogenetic\\ networks inferred from biological data}\label{sec:exp3}
The purpose of this case study is twofold: to confirm that the proposed methods remain fast and accurate on real data, and to illustrate that the base level is informative for analyzing the complexity of an inferred network.
We applied our methods to the two ARGs shown in Figs.~4(a) and (b) of~\cite{wong2024general}, which we call $N_{Kw}$ and $N_{weaver}$, respectively.
Both were inferred from the same dataset of $11$ \textit{Drosophila melanogaster} sequences at the alcohol dehydrogenase (\textit{Adh}) locus~\cite{kreitman1983nucleotide}, using KwARG~\cite{ignatieva2021kwarg} and ARGweaver~\cite{rasmussen2014genome}, respectively, and converted to the \texttt{tskit} graph-ARG representation. 
Ignoring the inheritance intervals carried by the edges, we treated both as rooted binary phylogenetic networks, and computed their support networks using the two algorithms of~\cite{SuzukiEtAl-2025-WhichPhylogeneticNetworks} and the two methods proposed in Section~\ref{sec:ILP}.
The implementations and the computational environment were the same as in Section~\ref{sec:exp}.

We first compare the runtimes in Table~\ref{tab:arg}. For $N_{Kw}$, Formulation~\ref{form:level-1 exact} performed best: it returned a level-1 support network, shown in Figure~\ref{fig:arg-a}, in $0.0042$ seconds, about half the time taken by either algorithm of~\cite{SuzukiEtAl-2025-WhichPhylogeneticNetworks}. 
For $N_{weaver}$, Formulation~\ref{form:level-1 exact} could only detect that 
its base level is more than one, but Formulation~\ref{form:level min} returned a level-$6$ support network, shown in Figure~\ref{fig:arg-b}, more than $1800$ times faster than the exact algorithm and about $8$ times faster than the heuristic of~\cite{SuzukiEtAl-2025-WhichPhylogeneticNetworks}. 
The base levels of $N_{Kw}$ and $N_{weaver}$ were $1$ and $6$, respectively, and both heuristics attained these values.

\begin{table*}[t]
\caption{Output base level and runtimes of each method for two rooted binary phylogenetic networks $N_{Kw}$ and $N_{weaver}$. 
For Formulation~\ref{form:level-1 exact}, ``YES'' indicates that there exists a level-1 support network, and ``NO'' indicates otherwise.
}
\label{tab:arg}
\centering
\scriptsize
\vspace{1em}
\begin{tabular}{@{}l rr rr rr rr@{}}\hline\hline
& \multicolumn{2}{c}{Algorithm~1 of~\cite{SuzukiEtAl-2025-WhichPhylogeneticNetworks}} & \multicolumn{2}{c}{Algorithm~2 of~\cite{SuzukiEtAl-2025-WhichPhylogeneticNetworks}}
& \multicolumn{2}{c}{Formulation~\ref{form:level-1 exact}} & \multicolumn{2}{c}{Formulation~\ref{form:level min}} \\
\cmidrule(lr){2-3}\cmidrule(lr){4-5}\cmidrule(lr){6-7}\cmidrule(lr){8-9}
Input network & base level & runtime & base level & runtime & base level~$\leq 1$ & runtime & base level & runtime \\
\midrule
$N_{Kw}$ (Fig.~4(a) in~\cite{wong2024general})    & 1 &   0.0079 & 1 & 0.0081 & YES & 0.0042 & 1 & 0.0107 \\
$N_{weaver}$ (Fig.~4(b) in~\cite{wong2024general}) & 6 & 292.8154 & 6 & 1.3915 & NO & 0.0113 & 6 & 0.1610 \\
\bottomrule
\end{tabular}
\end{table*}

\begin{figure}[t]
  \centering
  \includegraphics[width=0.7\textwidth]{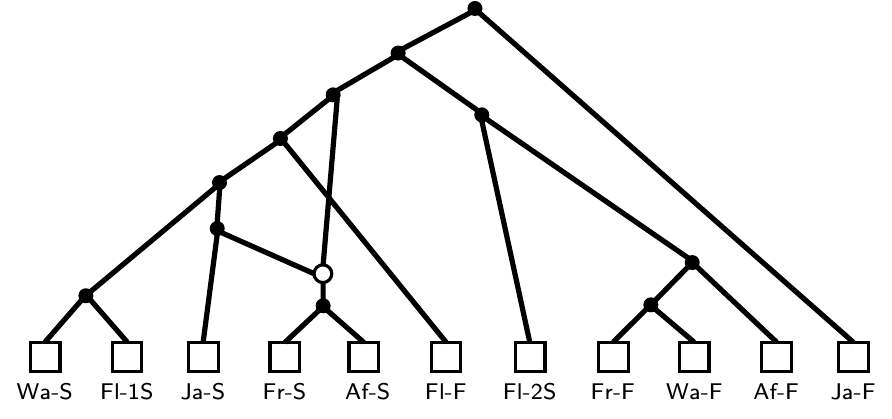}
  \caption{An output support network~$G$ of Formulation~\ref{form:level-1 exact} for the input network $N_{Kw}$ shown in Fig.~4(a) of~\cite{wong2024general}. Each edge is directed downwards, and every vertex of in-degree and out-degree one is smoothed for readability.
  }\label{fig:arg-a}
  \vspace{0.5em}
  \includegraphics[width=0.7\textwidth]{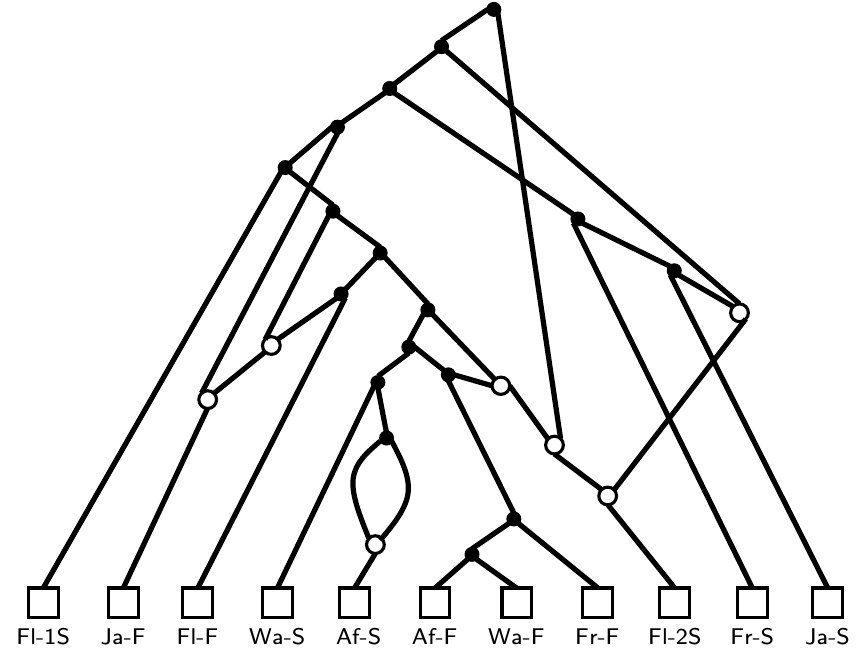}
  \caption{An output support network~$G$ of Formulation~\ref{form:level min} for the input network~$N_{weaver}$ shown in Fig.~4(b) of~\cite{wong2024general}. Each edge is directed downwards, and every vertex of in-degree and out-degree one is smoothed for readability.
  }\label{fig:arg-b}
\end{figure}

Wong et al.~\cite{wong2024general} pointed out that the complexity of the network estimated from the same data varies with the method used, and this is also captured by the base level: $N_{Kw}$ and $N_{weaver}$ contain $7$ and $37$ reticulations and have base levels $1$ and $6$, respectively. 
This suggests that the base level can be used as an alternative measure of the complexity of an inferred network. 
We note, however, that a support network of minimum level need not be unique, so that the networks shown in Figure~\ref{fig:arg-a} and Figure~\ref{fig:arg-b} cannot be claimed to represent the true evolutionary history.

It is also worth noting that the base level of $N_{weaver}$ is six, while Song and Hein~\cite{song2003parsimonious} showed that at least seven recombination events are required to explain the \textit{Adh} data under the infinite-sites model. 
Both quantities measure how far the data are from being explained by a single tree, one in terms of the entanglement of the reticulations in a topology and the other in terms of the number of recombination events in a history. 
The two cannot be identified with each other, since our methods use only the topology of an ARG, but the closeness of these values suggests a possible relationship between the base level of an inferred network and the parsimonious number of recombination events for the underlying data.

\section{Conclusions}\label{sec:conclusion}
In this paper, we have approached \textsc{Level Minimization} using integer linear programming.
We proposed an exact formulation for finding a support network of level at most one (Formulation~\ref{form:level-1 exact}), based on extracting pairwise vertex-disjoint reticulation cycles (Theorem~\ref{prop:level-1 chara}). 
By relaxing this disjointness, we also proposed a heuristic for \textsc{Level Minimization} (Formulation~\ref{form:level min}).
Our experiments showed that the runtime of Formulation~\ref{form:level-1 exact} grew polynomially with the size of the networks, and that Formulation~\ref{form:level min} ran faster on larger inputs than the heuristic of~\cite{SuzukiEtAl-2025-WhichPhylogeneticNetworks} at a slight cost in accuracy.
Furthermore, the runtime of Formulation~\ref{form:level min} was more stable among networks of the same size.
We note that the two ranges of $r$ used in Section~\ref{sec:exp} were determined by the design of the experiments, not by any limitation of the two formulations. 
An application to ARGs of \textit{Drosophila melanogaster}~\cite{wong2024general} illustrated that the base level can capture the difference in topological complexity between networks inferred from the same empirical data.

Several directions remain open. 
Formulation~\ref{form:level-1 exact} is stated for rooted almost-binary phylogenetic networks, and extending it to networks with vertices of degree four or more would widen its applicability. 
For Formulation~\ref{form:level min}, designing an objective function that better reflects the level of the resulting support network would improve its accuracy. An exact method for computing the base level would be more valuable still, but extending Formulation~\ref{form:level-1 exact} in this direction does not appear straightforward, as its constraints rely on the degree conditions of level-$1$ blocks and these do not readily generalize to level-$k$ blocks. 
Finally, on the biological side, our methods use only the topology of an ARG; incorporating the inheritance intervals carried by its edges would be necessary for the extracted support networks to be interpreted as evolutionary histories.

\section*{Acknowledgements}
The author thanks 
Martin Frohn for helpful discussions about ILP formulations and 
Momoko Hayamizu for helpful comments about biological interpretations. 
This work is part of the outcome of research performed under a Waseda University Grant for Special Research Projects (Project number: 2026C-088).

\bibliographystyle{unsrt}
\bibliography{takatora-counting}

@article{FrancisSteel-2015-WhichPhylogeneticNetworks,
  title = {Which {{Phylogenetic Networks}} Are {{Merely Trees}} with {{Additional Arcs}}?},
  author = {Francis, Andrew and Steel, Mike},
  year = {2015},
  journal = {Systematic Biology},
  volume = {64},
  number = {5},
  pages = {768--777},
  issn = {1063-5157, 1076-836X},
  doi = {10.1093/sysbio/syv037},
  urldate = {2024-11-26},
  langid = {english}
}

@article{Hayamizu-2021-StructureTheoremRooted,
  title = {A {{Structure Theorem}} for {{Rooted Binary Phylogenetic Networks}} and {{Its Implications}} for {{Tree-Based Networks}}},
  author = {Hayamizu, Momoko},
  year = {2021},
  journal = {SIAM Journal on Discrete Mathematics},
  volume = {35},
  number = {4},
  pages = {2490--2516},
  issn = {0895-4801, 1095-7146},
  urldate = {2024-11-25},
  langid = {english}
}

@article{Zhang-2016-TreeBasedPhylogeneticNetworks,
  title = {On {{Tree-Based Phylogenetic Networks}}},
  author = {Zhang, Louxin},
  year = {2016},
  journal = {Journal of Computational Biology},
  volume = {23},
  number = {7},
  pages = {553--565},
  issn = {1066-5277, 1557-8666},
  doi = {10.1089/cmb.2015.0228},
  urldate = {2025-05-07},
  copyright = {http://www.liebertpub.com/nv/resources-tools/text-and-data-mining-policy/121/},
  langid = {english}
}

@InProceedings{SuzukiEtAl-2025-WhichPhylogeneticNetworks,
  author =	{Suzuki, Takatora and Hayamizu, Momoko},
  title =	{{Which Phylogenetic Networks Are Level-$k$ Networks with Additional Arcs? Structure and Algorithms}},
  booktitle =	{25th International Conference on Algorithms for Bioinformatics (WABI 2025)},
  pages =	{19:1--19:19},
  ISBN =	{978-3-95977-386-7},
  ISSN =	{1868-8969},
  year =	{2025},
  volume =	{344},
  comment =	{Brejov\'{a}, Bro\v{n}a and Patro, Rob},
  publisher =	{Schloss Dagstuhl -- Leibniz-Zentrum f{\"u}r Informatik},
  address =	{Dagstuhl, Germany},
  URN =		{urn:nbn:de:0030-drops-239454},
  doi =		{10.4230/LIPIcs.WABI.2025.19}
}

@article{CHOY200593,
title = {Computing the maximum agreement of phylogenetic networks},
journal = {Theoretical Computer Science},
volume = {335},
number = {1},
pages = {93-107},
year = {2005},
issn = {0304-3975},
author = {Charles Choy and Jesper Jansson and Kunihiko Sadakane and Wing-Kin Sung}
}

@ARTICLE{Non-binary-tree-based,
  author={Jetten, Laura and van Iersel, Leo},
  journal={IEEE/ACM Transactions on Computational Biology and Bioinformatics}, 
  title={{N}onbinary {T}ree-{B}ased {P}hylogenetic {N}etworks}, 
  year={2018},
  volume={15},
  number={1},
  pages={205-217}
  }

@article{BAPTESTE2013439,
title = {Networks: expanding evolutionary thinking},
journal = {Trends in Genetics},
volume = {29},
number = {8},
pages = {439-441},
year = {2013},
issn = {0168-9525},
author = {Eric Bapteste and Leo {van Iersel} and Axel Janke and Scot Kelchner and Steven Kelk and James O. McInerney and David A. Morrison and Luay Nakhleh and Mike Steel and Leen Stougie and James Whitfield}
}

@article{szollHosi2015genome,
  title={Genome-scale phylogenetic analysis finds extensive gene transfer among fungi},
  author={Sz{\"o}ll{\H{o}}si, Gergely J and Dav{\'\i}n, Adri{\'a}n Arellano and Tannier, Eric and Daubin, Vincent and Boussau, Bastien},
  journal={Philosophical Transactions of the Royal Society B: Biological Sciences},
  volume={370},
  number={1678},
  year={2015},
  publisher={The Royal Society}
}

@article{goulet2017hybridization,
  title={{Hybridization} in {Plants}: {Old} {Ideas}, {New} {Techniques}},
  author={Goulet, Benjamin E and Roda, Federico and Hopkins, Robin},
  journal={Plant Physiology},
  volume={173},
  number={1},
  pages={65--78},
  year={2017},
  publisher={American Society of Plant Biologists}
}

@article{kong2022classes,
  title={Classes of explicit phylogenetic networks and their biological and mathematical significance},
  author={Kong, Sungsik and Pons, Joan Carles and Kubatko, Laura and Wicke, Kristina},
  journal={Journal of Mathematical Biology},
  volume={84},
  number={6},
  pages={47},
  year={2022},
  publisher={Springer}
}

@article{fss2018,
title = {New characterisations of tree-based networks and proximity measures},
journal = {Advances in Applied Mathematics},
volume = {93},
pages = {93-107},
year = {2018},
issn = {0196-8858},
doi = {https://doi.org/10.1016/j.aam.2017.08.003},
author = {Andrew Francis and Charles Semple and Mike Steel}
}

@article{francis2018tree,
  title={{T}ree-{B}ased {U}nrooted {P}hylogenetic {N}etworks},
  author={Francis, Andrew and Huber, Katharina T and Moulton, Vincent},
  journal={Bulletin of Mathematical Biology},
  volume={80},
  number={2},
  pages={404--416},
  year={2018},
  publisher={Springer}
}

@book{chartrand2024graphs,
  title={Graphs \& Digraphs},
  author={Chartrand, Gary and Jordon, Heather and Vatter, Vincent and Zhang, Ping},
  year={2024},
  edition={7th},
  publisher={Chapman and Hall/CRC}
}

@article{suzuki2026hardness,
  title={Recognizing Level-$k$-Based Phylogenetic Networks is NP-Complete},
  author={Suzuki, Takatora},
  journal={arXiv preprint arXiv:2605.26852},
  year={2026}
}

@article{wong2024general,
  title={A general and efficient representation of ancestral recombination graphs},
  author={Wong, Yan and Ignatieva, Anastasia and Koskela, Jere and Gorjanc, Gregor and Wohns, Anthony W and Kelleher, Jerome},
  journal={Genetics},
  volume={228},
  number={1},
  eid={iyae100},
  year={2024},
  publisher={Oxford University Press US}
}

@article{kreitman1983nucleotide,
  title={Nucleotide polymorphism at the alcohol dehydrogenase locus of \textit{Drosophila melanogaster}},
  author={Kreitman, Martin},
  journal={Nature},
  volume={304},
  number={5925},
  pages={412--417},
  year={1983},
  publisher={Nature Publishing Group UK London}
}

@article{ignatieva2021kwarg,
  title={KwARG: parsimonious reconstruction of ancestral recombination graphs with recurrent mutation},
  author={Ignatieva, Anastasia and Lyngs{\o}, Rune B and Jenkins, Paul A and Hein, Jotun},
  journal={Bioinformatics},
  volume={37},
  number={19},
  pages={3277--3284},
  year={2021},
  publisher={Oxford University Press}
}

@article{rasmussen2014genome,
  title={Genome-Wide Inference of Ancestral Recombination Graphs},
  author={Rasmussen, Matthew D and Hubisz, Melissa J and Gronau, Ilan and Siepel, Adam},
  journal={PLoS genetics},
  volume={10},
  number={5},
  eid={e1004342},
  year={2014},
  publisher={Public Library of Science San Francisco, USA}
}

@misc{gurobi,
  author = {{Gurobi Optimization, LLC}},
  title = {{Gurobi Optimizer Reference Manual}},
  year = 2026,
  url = "https://www.gurobi.com"
}

@InProceedings{song2003parsimonious,
author="Song, Yun S.
and Hein, Jotun",
editor="Benson, Gary
and Page, Roderic D. M.",
title="Parsimonious Reconstruction of Sequence Evolution and Haplotype Blocks",
booktitle="Algorithms in Bioinformatics",
year="2003",
publisher="Springer Berlin Heidelberg",
address="Berlin, Heidelberg",
pages="287--302",
isbn="978-3-540-39763-2"
}

@article{huson2011survey,
  title={A Survey of Combinatorial Methods for Phylogenetic Networks},
  author={Huson, Daniel H. and Scornavacca, Celine},
  journal={Genome Biology and Evolution},
  volume={3},
  pages={23--35},
  year={2011},
  publisher={Oxford University Press}
}

@article{van2009uniqueness,
  title={Uniqueness, intractability and exact algorithms: reflections on level-k phylogenetic networks},
  author={van Iersel, Leo and Kelk, Steven and Mnich, Matthias},
  journal={Journal of bioinformatics and computational biology},
  volume={7},
  number={04},
  pages={597--623},
  year={2009},
  publisher={World Scientific},
eprint = {https://doi.org/10.1142/S0219720009004308}
}

\end{document}